%% file: viveka.tex
\documentclass[conference]{IEEEtran}
\usepackage{cite}
\usepackage{amsmath,amssymb,amsfonts, amsthm}
\usepackage{soul}
\usepackage[ruled,linesnumbered]{algorithm2e}
\usepackage{subcaption}
\usepackage{graphicx}
\usepackage{textcomp}
\usepackage{xcolor}
\usepackage{svg}
\usepackage{url}
\usepackage{booktabs}
\usepackage{enumitem}

\newtheorem{theorem}{Theorem}
\newtheorem{remark}{Remark}

\def\BibTeX{{\rm B\kern-.05em{\sc i\kern-.025em b}\kern-.08em
    T\kern-.1667em\lower.7ex\hbox{E}\kern-.125emX}}
\begin{document}
    \input{sections/paper_meta}
    \input{sections/abstract}


    \input{sections/intro}

    \input{sections/motivation}
    \input{sections/system_model}
    \input{sections/viveka}
    \input{sections/evaluation}
    \input{sections/related}
    \input{sections/conclusion}

    \bibliographystyle{IEEEtran}
    \bibliography{refer}

    \input{sections/appendix.tex}

\end{document}

%% file: sections/paper_meta.tex
\title{
    Viveka: Context-Aware Sensing for Energy Efficiency in Smart Wearables
}


\author{
    \IEEEauthorblockN{Nikhil Sreekumar, Abhishek Chandra}
    \IEEEauthorblockA{\textit{Dept. of Computer Science and Engineering} \\
    \textit{University of Minnesota}\\
    MN, USA \\
    \{sreek012, chandra\}@umn,edu}
}

\maketitle

%% file: sections/abstract.tex
\begin{abstract}
    The proliferation of multi-sensor Internet of Things (IoT) systems, from Body Sensor Networks (BSNs) to industrial monitoring, is increasingly constrained by strict energy budgets and limited on-device storage. Continuous high-fidelity sensing leads to rapid battery depletion and data gaps that compromise application reliability. Existing strategies address this through sensor selection or adaptive sampling in isolation, or rely on computationally expensive agents for joint optimization. They lack context granularity or introduce significant overhead, and critically, they do not account for the risk that an aggressive, context-specific sensing policy applied to a misidentified context degrades accuracy. In this paper, we formulate joint sensor and sampling-rate selection as an NP-hard energy-minimization problem and propose Viveka \footnote{The name \textit{Viveka} is derived from the Sanskrit term for `discernment' or `discrimination' (distinguishing the real from the unreal). In this framework, it symbolizes the system's ability to semantically discern the \textbf{essential} (minimal sufficient sensors) from the \textbf{non-essential} (redundant data) based on the instantaneous context.}, a lightweight, context-aware framework. Viveka couples a cheap, always-on controller that estimates context and how much to trust that estimate with a stability and confidence gated policy that applies an aggressive per-context configuration only when context is certain, and falls back safely otherwise. Per-context configurations are instantiated using permutation feature importance and spectral energy analysis. Evaluation on the MHEALTH and PAMAP2 datasets shows that Viveka achieves up to 75\% energy savings and 78\% data reduction over standard baselines in a best-case configuration, while maintaining classification accuracy within 3-5\% of the baselines.
\end{abstract}


%% file: sections/intro.tex
\section{Introduction}
    The proliferation of Internet of Things (IoT) from ubiquitous smart cities to precision agriculture demonstrates how these multi-sensor systems are deeply embedded in our daily life. They drive applications ranging from daily fitness tracking \cite{bib:intro:fitness} and animal migration monitoring \cite{bib:intro:biolog} to high stakes scenarios such as cargo crew wellness assessment \cite{bib:intro:smartwatch}, military personnel health monitoring during active missions \cite{bib:intro:usnavy} and predictive maintenance in industrial manufacturing facilities \cite{bib:intro:industry}. In all these domains, the trend is towards increasing number of sensors, for instance in Body Sensor Networks (BSNs) or Industrial IoT (IIoT) clusters, to capture high fidelity data.

    \begin{figure}[h]
        \centering
        \includegraphics[scale=0.4]{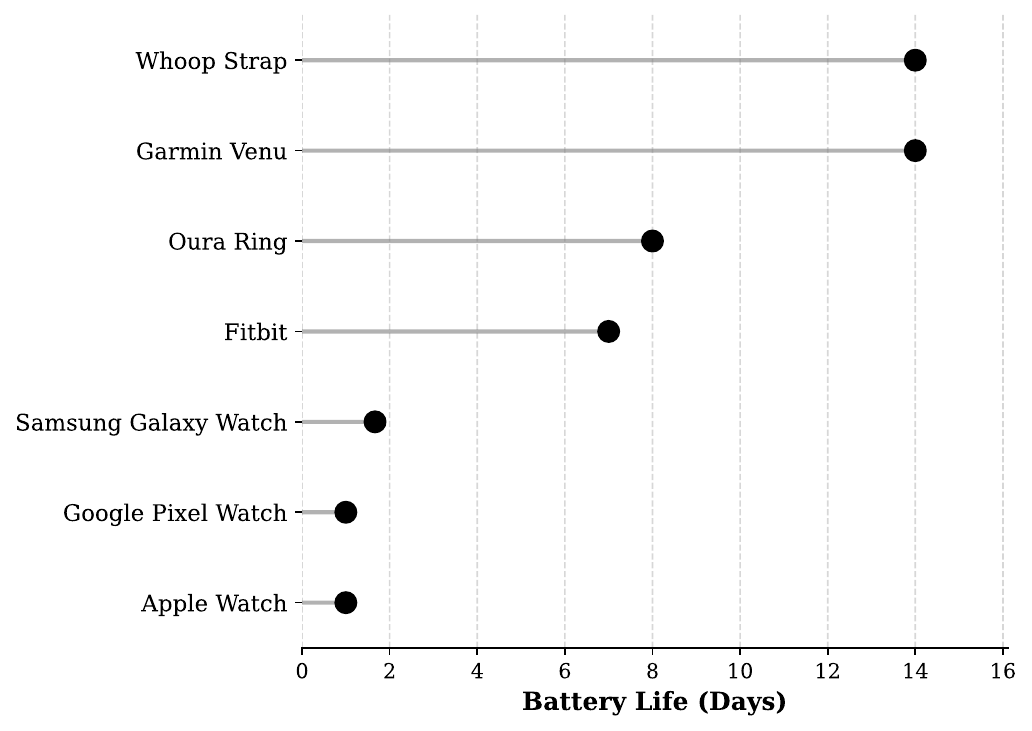}
        \caption{Battery of commercial smart wearables.}
        \label{fig:battery_life}
    \end{figure}
    The potential of these multi-sensor systems are impeded by the limited battery/ strict energy budgets, which has not kept pace with the growth of sensor capabilities and processing demands. Figure \ref{fig:battery_life} shows the battery life in days for known commercial smart wearables. Continuous high fidelity sensing is the worst case scenario for power consumption, leading to rapid battery depletion and frequent data gaps during recharge cycles. These gaps could result in critical event misses relevant to the application at hand. Furthermore, the deluge of data from always-on sensors overwhelms an already limited on-device storage on the sensors systems.

    To mitigate these constraints, existing research on energy optimization focus on two isolated strategies: Sensor selection and Adaptive sampling. In Sensor selection, a subset of sensors from a global set are identified to reduce active hardware components \cite{bib:intro:min_cost, bib:intro:sensor_class_co_opt, bib:intro:shapley, bib:intro:glowworm, bib:intro:vfds, bib:intro:dana}. As for Adaptive Sampling, the frequency of data collection per sensor is varied to lower the computation and sensing load \cite{bib:intro:context_aware, bib:intro:freqsense, bib:intro:lasa, bib:intro:ecotrack, bib:intro:adaptive_segments}. While both these approaches are effective, combining them could reap the compounding benefits. Existing joint optimization strategies typically rely on instance-wise reinforcement learning agents \cite{bib:related:idss}, early-exit neural architectures \cite{bib:intro:freqsense} or static model pruning \cite{bib:intro:coss}. However, these approaches often introduce significant operational penalties, including energy-inefficient sensor thrashing (sensors toggling ON/OFF), unpredictable computational overhead or an inability to adapt to instantaneous context changes.

    In this paper, we propose Viveka, a context-aware framework that achieves high energy savings by pruning data generation at the source, while remaining robust to the errors inherent in on-device context detection. The central difficulty is not identifying an efficient per-context sensing policy, but applying one safely: an aggressive, context-specific configuration applied to a misidentified context degrades recognition accuracy, a risk that static and globally-optimized policies do not address. Viveka resolves this with a closed-loop design. A lightweight, always-on controller continuously estimates the current context and how much to trust that estimate. A stability and confidence gated policy then applies the fully optimized per-context configuration only when context is certain, widens to a safe union of sensors under moderate uncertainty, and falls back to a conservative default when context is unresolved. The per-context configurations themselves are computed offline: the minimal sensor set per context via Permutation-based Feature Importance (PFI), and the lowest viable sampling rate per sensor via spectral energy analysis. The identified configurations are then stored in a granular, activity-aware lookup table. Because these configurations are exercised at runtime only under the gating policy's control, Viveka can exploit an aggressive policy without being derailed by an occasionally wrong context estimate.

    We validate Viveka using Human Activity Recognition (HAR) \cite{bib:intro:wear_survey} on smart wearables as a primary tested. HAR presents a challenging usecase due to the stochastic nature of human motion and the size/power constraints of wearable hardware. However, the principles of Viveka is generalizable to any energy constrained, multi-sensor IoT environment.

    The main contributions of this paper are as follows:
    \begin{itemize}
        \item We formulate sensor and sampling rate selection as an energy minimization problem constrained by classification error, and prove it is NP-hard, motivating a runtime-tractable heuristic rather than an exact solver.
        \item We propose Viveka, a closed-loop context-aware sensing framework built around a lightweight, always-on controller decoupled from the recognition model. The controller's role is not to classify accurately but to continuously estimate how much to trust the current context estimate, gating whether the aggressive, context-specific sensing policy can be safely applied.
        \item We introduce a stability and confidence gated decision policy with three operating tiers: (1) a fully optimized per-context policy when context is stable and confident, (2) a recovery policy unioning the sensor requirements of plausible contexts under moderate uncertainty, and (3) a conservative fallback when context is unresolved.
        \item We co-design training with deployment: the recognition model is trained with rate-degraded, reconstructed-signal augmentation so it remains robust to the subsampled input it receives at runtime, closing the gap between offline policy computation and online execution.
        \item We instantiate per-context sensor and rate selection using permutation feature importance and spectral (Nyquist-based) energy analysis. We then evaluate the full system under a realistic, delayed predict-and-act protocol with a datasheet-grounded energy model spanning sensing, compute and radio.
        \item Extensive evaluation on MHEALTH and PAMAP2 demonstrates that Viveka achieves up to 75\% energy savings and 78\% data reduction over standard baselines, \footnote{These figures reflect a best-case configuration in which policy thresholds were selected on a single subject-disjoint split. Subsequent leave-one-subject-out cross-validation, with thresholds selected only on held-out training subjects, yields more conservative savings (approximately 64\% on MHEALTH and 49\% on PAMAP2) while preserving comparable accuracy; a full treatment is the subject of forthcoming work.} while maintaining classification accuracy within 3-5\% of the baselines.
    \end{itemize}

%% file: sections/motivation.tex
\section{Motivation}
\label{sec:motive}
    In this section, we explore various wearable applications that benefit from extended battery life and the application-specific insights that can be leveraged to achieve it.

    \subsection{Motivating Applications}
    The need for extended battery life is essential to realize the full potential of wearable technology. Consider the following applications where extended battery life is critical:
    \begin{itemize}
        \item \textbf{Daily Fitness and Wellness Tracking}: The consumer wellness market represents the most widespread use of wearable technology. Devices that track daily fitness, sleep quality and stress levels rely on consistent, long-term use to provide meaningful insights. The user experience is directly tied to the device's convenience and a primary point of friction is the need for frequent charging. A user who must charge their watch every day may forget to put it back on, resulting in incomplete data for activity goals and a total loss of sleep tracking for that night. To become a seamless part of a user's life, these devices must demand minimal attention. Extended battery life is therefore a key feature that directly impacts user retention and the overall effectiveness of the wellness platform \cite{bib:motive:health_fitness}.

        \item \textbf{Crew Readiness Monitoring in Critical Missions}: Monitoring the operational readiness of personnel in austere environments, such as on oil rigs or during military missions, is critical for safety and mission success. In these remote and demanding settings, frequent charging is not just inconvenient but also introduces a significant point of failure. A device with a depleted battery cannot provide vital alerts on fatigue or stress, directly compromising crew safety. To be effective, these wearables must operate for extended periods with minimal user intervention. This necessitates advanced, on-board energy management systems that can balance data collection with power preservation to ensure reliability \cite{bib:intro:smartwatch}.
    \end{itemize}

    Across these applications, activity recognition is a common thread. In all the cases, the primary goal is to capture and analyze activity patterns to extract meaningful insights. However, the energy constraints of wearable devices make this goal challenging to achieve. To address this challenge, researchers have proposed various techniques to optimize energy consumption while maintaining the accuracy of the activity recognition system, including the sensor selection and adaptive sampling rate methods discussed previously.

    \subsection{Sensor energy-accuracy trade offs}
        The effectiveness of sensor selection arises from the significant variance in power consumption across different sensors. Consider a typical set of commercial sensors used in a smartwatch \cite{bib:motive:bosch_wearable}, with specifications taken from their datasheets: the BMI270 (accelerometer and gyroscope), BMM350 (magnetometer) and BMP858 (pressure sensor). As shown in Figure \ref{fig:motive:current_sensors}, the current consumption for these sensors operating at a comparable low-power sampling rate (near 50 Hz) varies dramatically. The gyroscope is the most power-hungry, followed by the magnetometer, pressure sensor and accelerometer. In this scenario, if the gyroscope's data could be omitted for certain activities without a significant loss in HAR accuracy, the current consumption for sensing could be reduced by approximately 70\%. This highlights the need to identify the optimal set of sensors required for each activity to maximize energy savings.

        \begin{figure}
            \centering
            \includegraphics[scale=0.5]{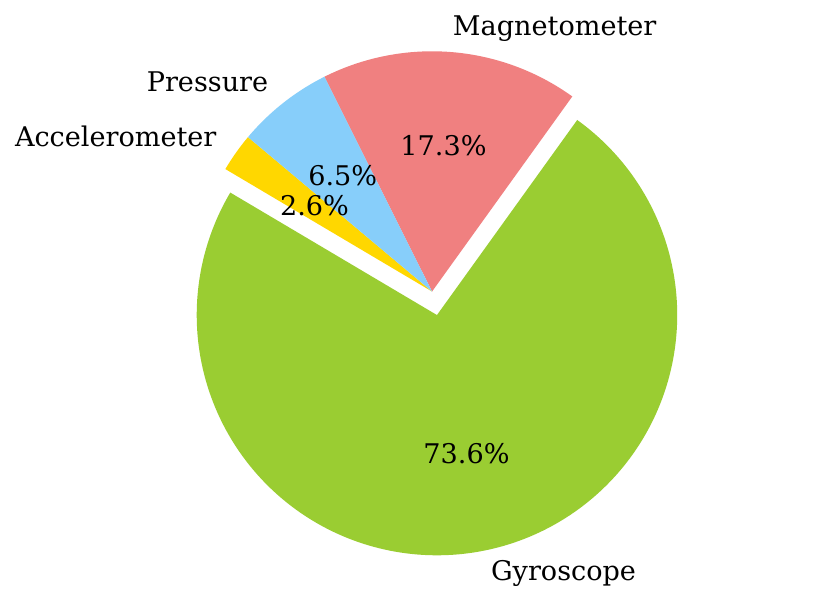}
            \caption{Current consumption of sensors at $~$50Hz}
            \label{fig:motive:current_sensors}
        \end{figure}

        To gain better insight into the relationship between activities and sensors, we conducted an experiment where a Long Short-Term Memory (LSTM) model was trained on the WISDM \cite{bib:motive:wisdm} dataset. The WISDM dataset consists of data from smartphones and smartwatches across multiple users for 18 activities. Since our focus is on wearables, only the smartwatch data was used for training. The data consists of accelerometer and gyroscope data sampled at 20 Hz. We then applied Permutation Feature Importance (PFI) \cite{bib:motive:pimp} to analyze the contribution of each sensor feature to each activity classification. Figure \ref{fig:motive:pfi_20hz} shows the accuracy drop when a feature is permuted, indicating the relative importance of each sensor for each activity. The x-axis represents the sensor features (accelerometer: acc\_x, acc\_y, acc\_z; gyroscope: gyro\_x, gyro\_y, gyro\_z) and the y-axis represents the activities (A to S). The results show that the importance of each sensor varies by activity. For example, activity \textit{A} is dependent on both accelerometer and gyroscope data, while activity \textit{D} is mostly dependent on accelerometer data. This insight can be utilized at runtime to decide which sensors to activate for a given activity, thereby reducing energy consumption.

        \begin{figure}
            \centering
            \resizebox{0.5\textwidth}{!}{
                \includegraphics[scale=0.4]{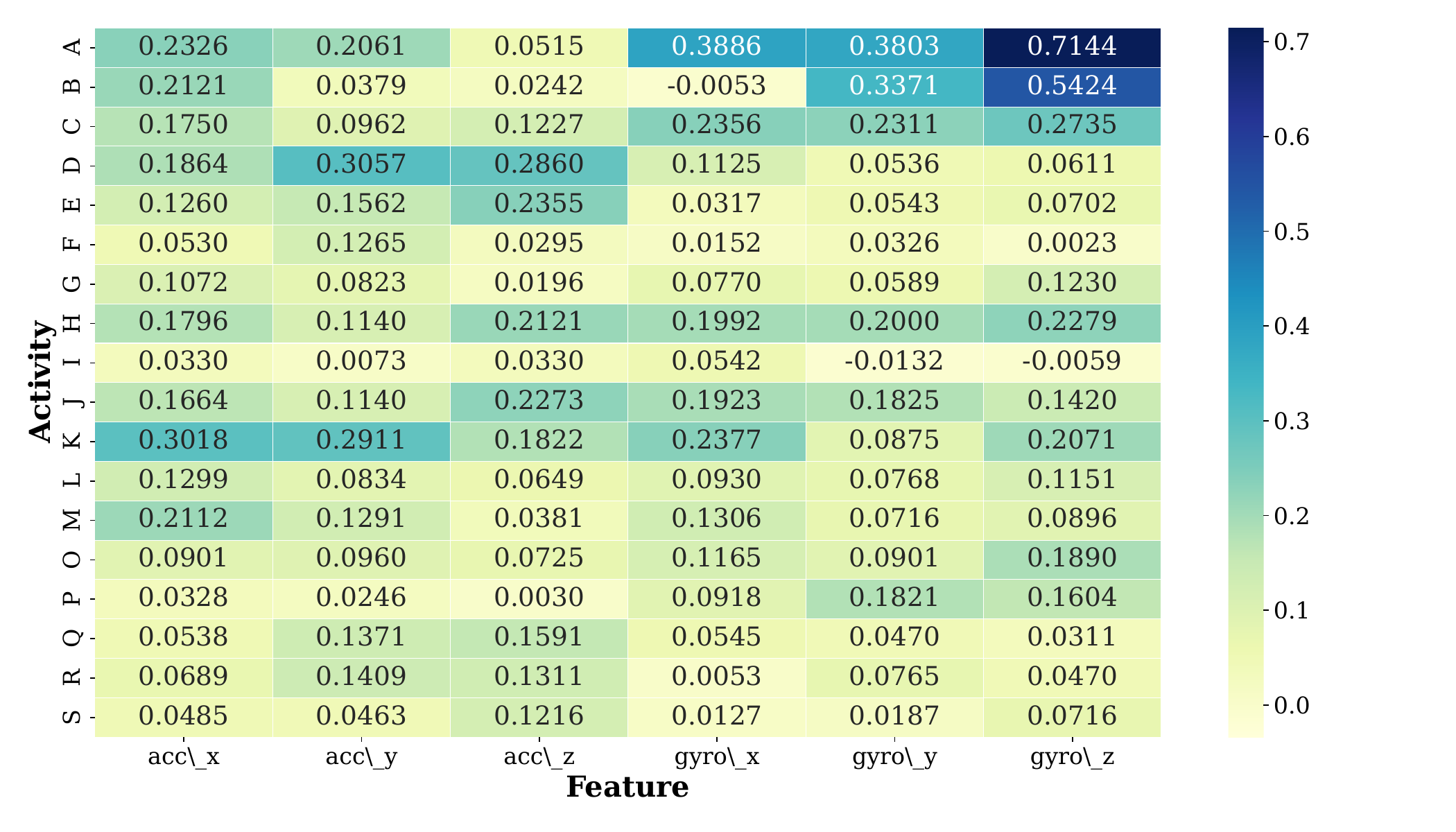}   
            }
            \caption{PFI for the WISDM dataset at 20Hz}
            \label{fig:motive:pfi_20hz}
        \end{figure}
            
        \begin{figure}
            \centering
            \resizebox{0.5\textwidth}{!}{
                \includegraphics[scale=0.4]{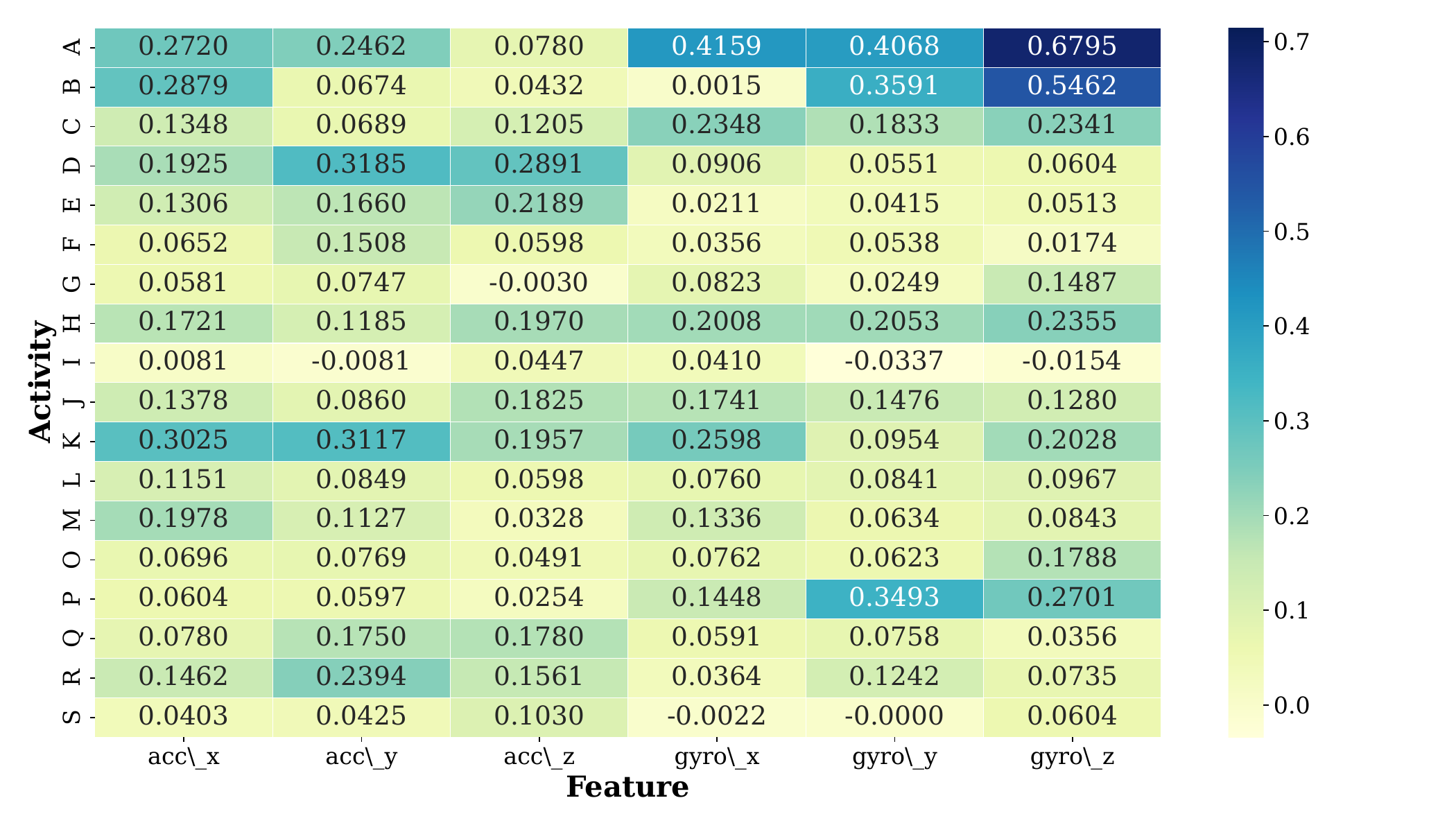}
            }
            \caption{PFI for the WISDM dataset at 10Hz}
            \label{fig:motive:pfi_10hz}
        \end{figure}

        The sensor sampling rate also provides an opportunity for energy savings, as it can often be reduced without a major impact on performance. To demonstrate this, we extended our previous experiment using the same LSTM model trained on the 20 Hz WISDM dataset. During testing, we downsampled the data to mimic 10 Hz sensing by selecting every other data point. We then interpolated this data back to 20 Hz before feeding it to the model. This resulted in the overall test accuracy dropping by only 3.66\%, from 84.11\% to 80.44\%. Figure \ref{fig:motive:pfi_10hz} shows the permutation-based accuracy drop for the 10 Hz data. When compared to the original 20 Hz results, there is only a slight change in the accuracy drop for most activities, with the notable exception of gyro\_z for activity \textit{A}. In addition, it could be also see that the drop in accuracy slightly increases in 10 Hz. This is due to the interpolation process resulting in smoother data. This shows that even at a lower sampling rate, the model can still achieve high accuracy, highlighting the potential of adaptive sampling to reduce energy consumption while maintaining classification performance.

%% file: sections/system_model.tex
\section{System Model}
    In this section, we describe the system model for our work. Based on this model, we formulate the sensor and sampling rate selection as an energy minimization problem.
    \subsection{System Components}
        The system we consider consists of multiple sensors that reside on a single edge device, such as a smartwatch, ring, or band. Alternatively, the system could be a body sensor network (BSN) where sensors connect to a central gateway device. The edge device is constrained by limited battery capacity, processing power (CPU/GPU), and storage. It is powered by a battery of limited capacity, necessitating periodic recharging. Data collected from the sensors is processed and stored on the device for a period before being transmitted to a nearby edge server (e.g., a smartphone, tablet, or hub) via low-energy communication protocols such as Bluetooth Low Energy (BLE). Upon receiving the data, the edge server reconstructs the original signal using techniques like linear interpolation. This reconstructed data is then fed to a model for activity recognition and the results are utilized by upstream applications. Figure \ref{fig:sysmodel:system_model} shows a high-level overview of the system.
        
        \begin{figure}
            \centering
            \includegraphics[scale=0.4]{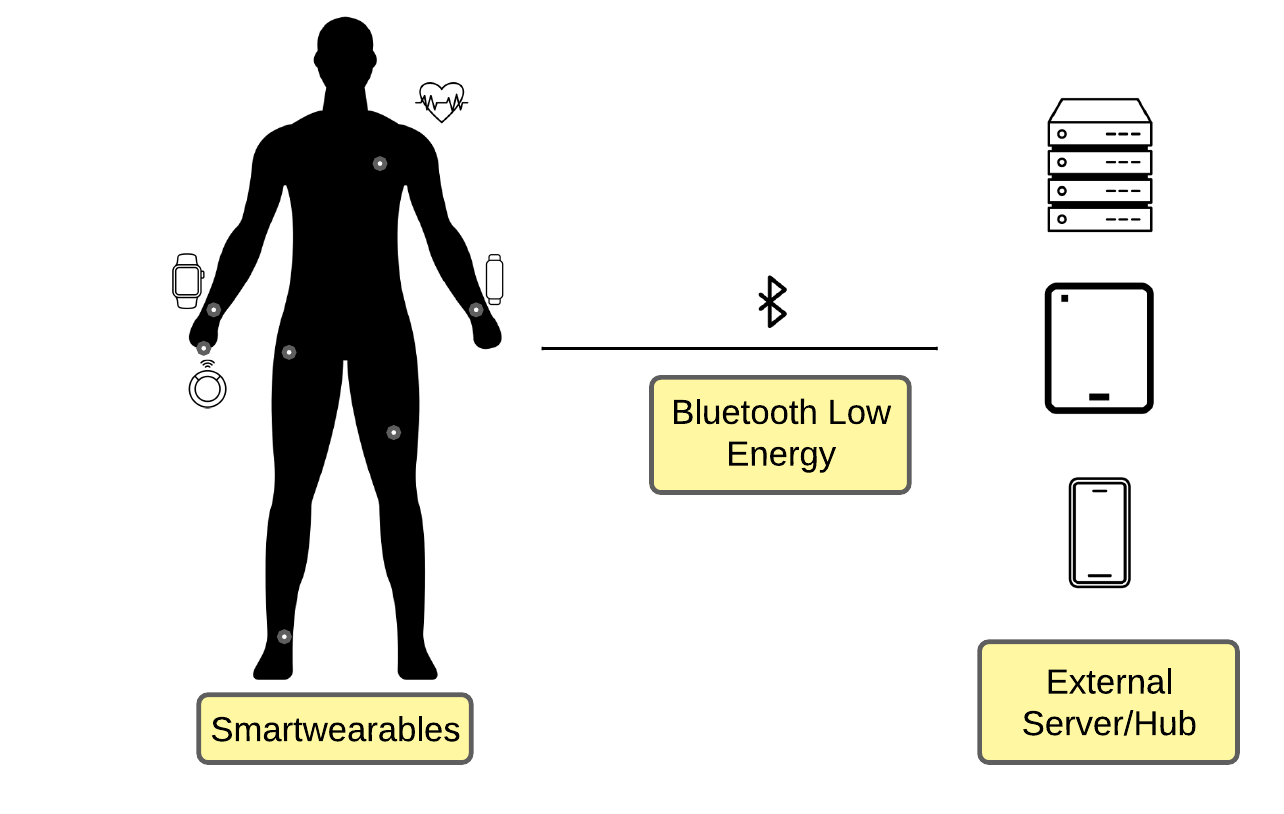}
            \caption{System Model}
            \label{fig:sysmodel:system_model}
        \end{figure}

    \subsection{Problem Formulation}
        We will now formulate the sensor selection and adaptive sampling as an optimization problem based on the system model.

        Consider the following parameters for the formulation:
            \begin{itemize}[leftmargin=8pt]
                \item $\mathcal{S}$: The set of available sensors.
                \item $R_{s}$: The discrete set of sampling rates for sensor $s \in \mathcal{S}$, indexed by $j$. To account for sensor selection, we explicitly include the zero rate (sensor OFF). Thus, $R_s = \{r_{s, 0}, r_{s, 1}, \dots, r_{s, |R_s|-1}\}$, where $r_{s, 0} = 0$ represents the inactive state and $r_{s, j} > 0$ for $j > 0$ represents active sampling rates.
                \item $E(s, r)$: The power consumption of sensor $s$ operating at sampling rate $r$. We assume $E(s, 0) = 0$.
                \item $\mathcal{X}$: The set of all feasible configurations. A configuration vector $x \in \mathcal{X}$ represents the sampling rate selection for all sensors, denoted as $x = [x_{1}, x_{2}, \cdots, x_{|\mathcal{S}|}]$, where $x_{s} \in R_{s}$ is the rate chosen for sensor $s$. The size of this search space is given by $|\mathcal{X}| = \prod_{s \in \mathcal{S}} |R_s|$. If all sensors have $N$ rates, this simplifies to $N^{|\mathcal{S}|}$, indicating exponential growth with respect to the number of sensors.
                \item $C(x)$: The expected classification error of the inference model when data is collected using configuration $x$.
                \item $C_{threshold}$: The maximum tolerable classification error for the application to remain viable.
            \end{itemize}

        \subsubsection{Energy and Decision Variables}

            The total energy consumption for a configuration $x$ is the sum of the energy costs of the selected rates for each sensor:
            \begin{align}
                E_{total} (x) = & \sum_{s \in \mathcal{S}} E(s, x_{s})
            \end{align}

            To formulate the optimization problem as an Integer Linear Program (ILP), we introduce binary decision variables $\delta_{s, j}$ to represent the selection of specific rates:
            \begin{align}
                \delta_{s, j} = \begin{cases}
                            1  & \text{if sensor $s$ is assigned sampling rate $r_{s, j}$}\\
                            0  & \text{otherwise}
                            \end{cases}
            \end{align}

            These binary variables map to the configuration $x$ such that $x_{s} = \sum_{j} r_{s,j} \cdot \delta_{s,j}$. Consequently, the energy function can be rewritten in terms of $\delta$:
            \begin{align}
                E_{total} (x) = \sum_{s \in \mathcal{S}} \sum_{j} E(s, r_{s, j}) \cdot \delta_{s, j}
            \end{align}

        \subsubsection{Optimization Problem (APSSE)}
            We will call the optimization as \textit{Accuracy Preserving Selection and Sampling for Energy efficiency} (APSSE) problem.

            \begin{align}
                \textbf{Minimize:} \quad & \sum_{s \in S} \sum_{j} E(s, r_{s, j}) \cdot \delta_{s, j} \label{eq:min:obj}\\
                \textbf{Subject to:} \quad & C(x) \leq C_{threshold} \label{eq:cons:1}\\
                                        & \sum_{j} \delta_{s, j} = 1, \quad \forall s \in S \label{eq:cons:2}\\
                                        & \delta_{s, j} \in \{0, 1\}, \quad \forall s \in S, \forall j \label{eq:cons:3} 
            \end{align}
            The summations over $j$ iterate through all valid rate indices $\{0, \dots, |R_s|-1\}$ for sensor $s$. The constraint in \eqref{eq:cons:1} depends on the configuration $x$, which is determined by the decision variable $\delta$ as described above. Constraint \eqref{eq:cons:1} ensures the solution meets the accuracy requirement. Constraint \eqref{eq:cons:2} enforces that every sensor must be in exactly one state (which includes the state of being OFF/0Hz).

            In Appendix \ref{apdx:nphard}, we prove APSSE is NP-Hard. The computational complexity of the APSSE is $O(N^{|\mathcal{S}|})$, where $N$ is the number of discrete sampling rates per sensor (we assume it is same across all sensors for ease) and $|\mathcal{S}|$ is the number of sensors. It is exponential in the number of sensors. For a system with large number of sensors, an exact solver will require a long time to converge, rendering it completely unsuitable for online, per-activity adaptation, thereby necessitating the heuristic approach proposed in the next section.

%% file: sections/viveka.tex
\section{Viveka}
    In this section, we propose Viveka, a novel, context-aware framework that selects sensors and corresponding sampling rates based on the context to reduce the energy consumption associated with sensing. Since the APSSE joint optimization is NP-Hard, Viveka employs a two stage heuristic: first, it determines the minimal set of sensors required per activity. Second, it identifies the sample rate required per sensor in the minimal set while ensuring minimal impact on overall task accuracy. 

    Figure \ref{fig:viveka:components} showcases the Viveka framework and its main modules. Our primary contributions lie within the Sensor Selection and Adaptive Sampling modules.

        \begin{figure}
            \centering
            \resizebox{0.5\textwidth}{!}{
                \includegraphics[scale=0.6]{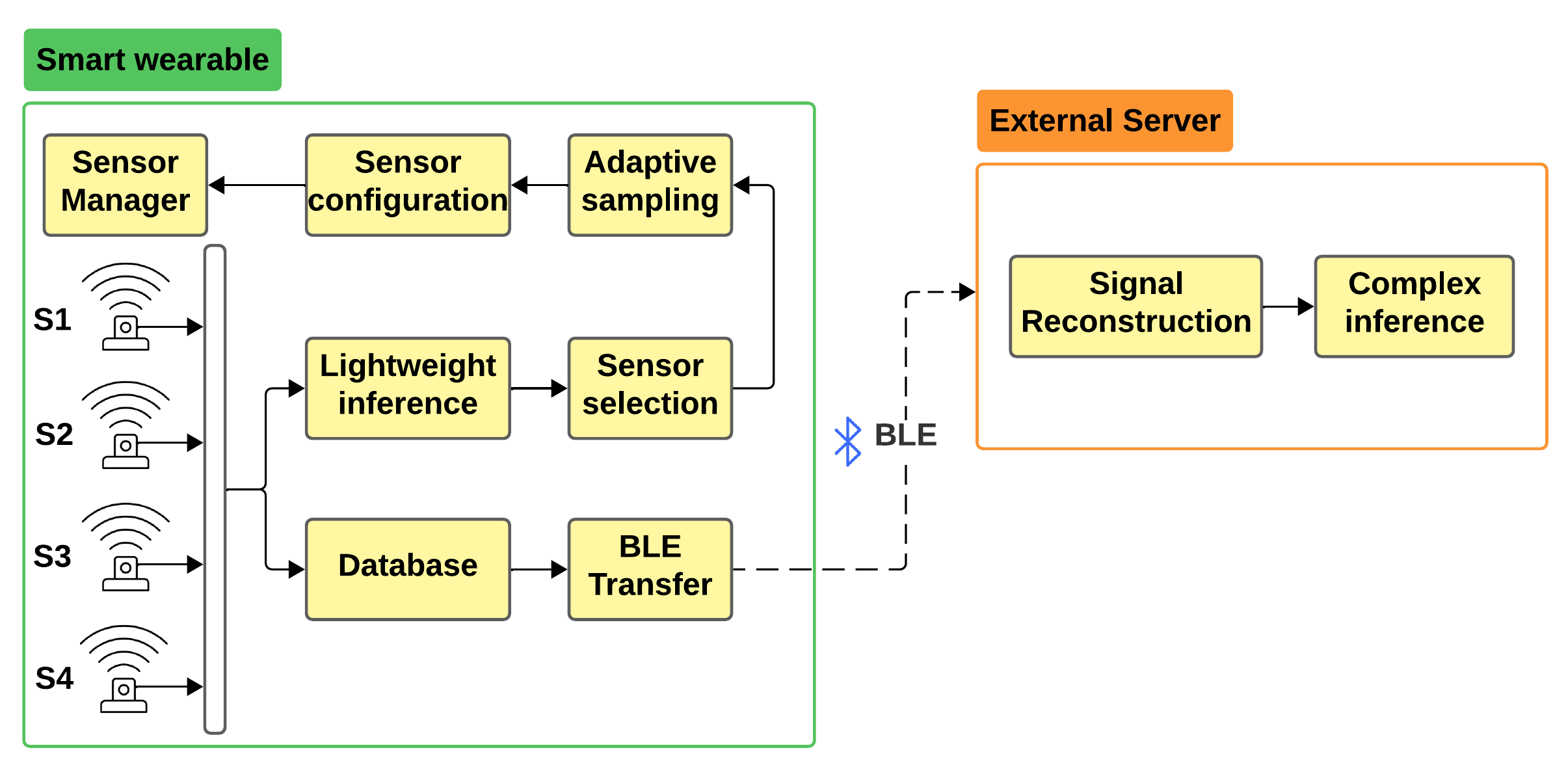}  
            }
            \caption{Viveka framework}
            \label{fig:viveka:components}
        \end{figure}
    
    \subsection{Sensor Selection}
        The strategy (Algorithm \ref{alg:offline_sensor_selection}) employs an offline analysis to determine the most important sensor set for each distinct activity using Permutation-based Feature Importance (PFI). In our system, we have two pre-trained models: A light-weight model is used for real-time activity inference on the wearable device, while the model on edge server is used for more complex processing. The goal of the light-weight model is to provide an understanding of the current context and we do not expect perfect accuracy. For example, it may report someone is doing an activity similar to running or jogging. To ensure a comprehensive selection, PFI is calculated independently for the two models. In PFI, the features corresponding to a sensor are permuted and fed to the model to measure the impact on accuracy. A significant drop in model accuracy indicates that the sensor's features are highly important. For each activity, we identify the sensors that are important for both models to ensure that no critical information is lost when the data is processed on the edge server. The selection of these sensors is based on a PFI threshold ($\theta_{PFI}$). The final set of essential sensors for a given activity is the union of the important sensors identified from both models, creating a policy ($\Pi_{\mathcal{S}}$) that is effective across the entire system.

        \begin{algorithm}[t]
            \footnotesize
            \caption{Offline Generation of the Sensor Selection Map ($\Pi_{\mathcal{S}}$)}
            \label{alg:offline_sensor_selection}
            \SetKwInOut{Input}{Input}
            \SetKwInOut{Output}{Output}
            \SetAlgoLined
            \DontPrintSemicolon

            \Input{
                Training data $\mathcal{D}$, models $M_{wearable}$ \& $M_{edge}$, activities $\mathcal{A}$, sensors $\mathcal{S}$, PFI threshold $\theta_{PFI}$
            }
            \Output{Policy map $\Pi_{\mathcal{S}}$ mapping each activity to its optimal sensor set.}
            \BlankLine

            \tcp{Initialize the policy map}
            $\Pi_{\mathcal{S}} \leftarrow \text{new empty Map()}$

            \ForEach{activity $A \in \mathcal{A}$}{
                \tcp{Filter data for the current activity}
                $D_A \leftarrow \{d \in \mathcal{D} \mid \text{label}(d) = A\}$
                
                \tcp{Identify important sensors from both models using PFI}
                $S_{w} \leftarrow \{s \in \mathcal{S} \mid \text{PFI}(M_{wearable}, D_A, s) \ge \theta_{PFI}\}$\\
                $S_{e} \leftarrow \{s \in \mathcal{S} \mid \text{PFI}(M_{edge}, D_A, s) \ge \theta_{PFI}\}$

                \tcp{Store the union of important sensors}
                $\Pi_{\mathcal{S}}[A] \leftarrow S_{w} \cup S_{e}$
            }
            \Return $\Pi_{\mathcal{S}}$
        \end{algorithm}
    
    \subsection{Per sensor Sampling rate Selection}
        Following the identification of essential sensors ($\Pi_{\mathcal{S}}$) for each activity, the strategy determines the optimal sampling rate for each of these sensors individually (Algorithm \ref{alg:offline_sampling_rate}). This is accomplished by applying the Nyquist-Shannon sampling theorem in an energy-aware context. For a specific activity, the signal data from each essential sensor is transformed into the frequency domain using a Fast Fourier Transform (FFT). The analysis then identifies the maximum frequency component needed to capture a predefined portion of the signal's total energy, as determined by an energy threshold ($\theta_{FFT}$). According to the Nyquist-Shannon theorem, the minimum sampling rate required to reconstruct the signal without losing this information is twice this maximum frequency. The strategy then selects the lowest available discrete sampling rate from a predefined list (e.g., [1, 2, 5, 10, 25, 50] Hz) that meets or exceeds this calculated minimum ($R_{map, A}[s]$). This per-sensor optimization results in an efficient policy ($\Pi_{\mathcal{R}}$) where, for the same activity, a high-frequency sensor like an accelerometer may be sampled at 50 Hz while a slower-changing sensor like a magnetometer is sampled at only 10 Hz, significantly reducing overall data volume and energy consumption. The tuning of the FFT energy threshold ($\theta_{FFT}$) allows for a trade-off between data quality and energy consumption.

            \begin{algorithm}[t]
                \footnotesize
                \caption{Offline Generation of the Optimal Sampling Rate Map ($\Pi_{\mathcal{R}}$)}
                \label{alg:offline_sampling_rate}
                \SetKwInOut{Input}{Input}
                \SetKwInOut{Output}{Output}
                \SetAlgoLined
                \DontPrintSemicolon

                \Input{
                    Training data $\mathcal{D}$, Sensor map $\Pi_{\mathcal{S}}$, Activities $\mathcal{A}$, FFT energy threshold $\theta_{FFT}$
                }
                \Output{Optimal sampling rate map $\Pi_{\mathcal{R}}$.}
                \BlankLine

                \tcp{Initialize the rate policy map}
                $\Pi_{\mathcal{R}} \leftarrow \text{new empty Map()}$

                \ForEach{activity $A \in \mathcal{A}$}{
                    \tcp{Get important sensors for the activity}
                    $S_{A} \leftarrow \Pi_{\mathcal{S}}[A]$
                    
                    \tcp{Initialize a map for the current activity's rates}
                    $R_{map, A} \leftarrow \text{new empty Map()}$
                    
                    \ForEach{sensor $s \in S_{A}$}{
                        \tcp{Find max frequency and calculate Nyquist rate}
                        $f_{max} \leftarrow \text{FFT}(\text{Signal}(\mathcal{D}, A, s, \theta_{FFT}))$\\
                        $R_{map, A}[s] \leftarrow 2 \times f_{max}$
                    }
                    
                    $\Pi_{\mathcal{R}}[A] \leftarrow R_{map, A}$
                }
                \Return $\Pi_{\mathcal{R}}$
            \end{algorithm}
    
    \subsection{Real time application of Sensor Selection and Sampling}
        During real-time operation, the system combines knowledge from its offline analysis with live data to make dynamic sensing decisions (Algorithm \ref{algo:real_time_sensing}). The wearable model periodically infers the user's current activity ($A_{pred}$) based on a window of sensor data. This inference frequency is governed by the Decision Interval ($T_{interval}$), a tunable parameter balancing responsiveness and computation. An activity is confirmed if its prediction confidence exceeds a strict confidence threshold ($\theta_{conf}$).

        The system maintains a short-term history ($H$) of these predictions to assess the stability of the current activity. If the last few predictions are consistent, the system assumes a stable activity and enforces the specific pre-computed policy ($\Pi_{\mathcal{S}}, \Pi_{\mathcal{R}}$) for that activity. If the system is unstable and multiple activities are detected at the moment, then we use a fallback confidence threshold ($\beta_{conf}$) to select a subset of activities and the policies corresponding to those activities are enforced.
        
        Finally, if the system in unstable and no activity predictions are above $\beta_{conf}$, the system enters a high-fidelity safe mode ($C_{safe}$).  In this mode, it activates a default set of sensors at an average sample rate to ensure the new activity is accurately identified. This default set is determined beforehand by applying PFI to the entire training dataset. Conversely, once predictions stabilize for a predefined duration, the system consults the pre-computed policy map. It then retrieves and applies the specific set of essential sensors ($\Pi_{\mathcal{S}}$) and their optimized sampling rates ($\Pi_{\mathcal{R}}$) corresponding to the stable activity. This policy remains active until the next transition is detected. This dual-mode approach ensures maximum data quality during periods of uncertainty and maximum energy efficiency during periods of stability.

            \begin{algorithm}[t]
                \footnotesize
                \caption{Real-Time Adaptive Sensing}
                \label{algo:real_time_sensing}
                \SetKwInOut{Input}{Input}
                \SetAlgoLined
                \DontPrintSemicolon

                \Input{
                    Policies $\Pi_{\mathcal{S}}$, $\Pi_{\mathcal{R}}$; Thresholds $\theta_{conf}, \beta_{conf}$; Safe Config $C_{safe}$; Wearable model $M_{w}$
                }
                \BlankLine

                \tcp{Initialization}
                $H \leftarrow \text{new Queue()}$ \tcp{History buffer}
                $C_{current} \leftarrow C_{safe}$

                \While{true}{
                    \tcp{Wait for next Decision Interval}
                    $\text{Wait}(T_{interval})$
                    
                    $D_{window} \leftarrow \text{CollectData}(C_{current})$
                    $probs \leftarrow M_{w}.\text{Predict}(D_{window})$
                    $A_{top} \leftarrow \text{argmax}(probs)$
                    
                    \tcp{Update History}
                    $H.\text{Enqueue}(A_{top})$
                    
                    \If{$\text{IsStable}(H) \textbf{ and } probs[A_{top}] \ge \theta_{conf}$}{
                        \tcp{Stable State: Apply specific optimal policy}
                        $C_{current} \leftarrow \text{Config}(\Pi_{\mathcal{S}}[A_{top}], \Pi_{\mathcal{R}}[A_{top}])$
                    }
                    \Else{
                        \tcp{Unstable: Check for Fallback Candidates}
                        $Candidates \leftarrow \{c \in \mathcal{A} \mid probs[c] \ge \beta_{conf}\}$
                        
                        \If{$Candidates \neq \emptyset$}{
                            \tcp{Fallback State: Union of likely policies}
                            $S_{union} \leftarrow \bigcup_{c \in Candidates} \Pi_{\mathcal{S}}[c]$
                            $R_{union} \leftarrow \text{AverageRateMap}(Candidates, \Pi_{\mathcal{R}})$
                            $C_{current} \leftarrow \text{Config}(S_{union}, R_{union})$
                        }
                        \Else{
                            \tcp{Unknown State: Revert to Safe Mode}
                            $C_{current} \leftarrow C_{safe}$
                        }
                    }
                    $\text{Apply}(C_{current})$
                }
            \end{algorithm}
    
    \subsection{Server-Side Inference}
        Data collected by the wearable according to the dynamically adapted policy is periodically transferred to an edge server for final, high-accuracy activity inference using the more powerful edge model (Figure \ref{fig:viveka:components}). This data transmission, orchestrated via Bluetooth Low Energy (BLE), can be triggered at regular intervals or when the on-device memory buffer approaches its capacity. The energy-efficient sensing strategies offer a dual benefit: in addition to reducing sensor power consumption, they also decrease the data payload that must be transmitted, which in turn lowers the total energy cost of communication. Upon receiving the sparsely-sampled data, the edge server first performs a reconstruction step, upsampling the signal from each sensor back to the original base sampling rate via signal interpolation. This reconstructed data is then fed into the edge model for final activity classification.

%% file: sections/evaluation.tex
\section{Evaluation}
    \subsection{Experiment Setup}
        This section details the methodology used to evaluate Viveka. The experiments are designed to answer the following questions:
        \begin{itemize}[leftmargin=8pt]
            \item How significant are the energy savings achieved by our method compared to state-of-the-art approaches?
            \item What is the trade-off between the achieved energy savings and the resulting HAR performance?
        \end{itemize}

        \emph{Implementation Details}:
        The simulation framework is developed in Python. We utilized the scikit-learn library for data preprocessing and employed TensorFlow and TensorFlow Lite (TFLite) to train and optimize models for on-device inference on resource-constrained hardware. All models were trained on a desktop system running Ubuntu 24.04.3 LTS, equipped with a 13th Gen Intel Core i7-13700HX processor, 16 GB of RAM, and an NVIDIA GeForce RTX 4060 GPU.

        \emph{Dataset}:
        Our evaluation uses two real world datasets (Table \ref{tab:datasets}):
        \begin{itemize}[leftmargin=8pt]
            \item MHEALTH \cite{bib:eval:mhealthdroid,bib:eval:mhealth}: It comprises recordings from 10 volunteers performing 12 physical activities and features data from inertial sensors (accelerometer, gyroscope, magnetometer) placed on the subject's chest, right wrist and left ankle. A two-lead ECG sensor was also placed on the chest. The dataset contains data from eight sensors in total, all of which were sampled at 50 Hz.
            \item PAMAP2 \cite{bib:eval:pamap2}: It consists of recordings from 9 volunteers performing 18 activities wearing sensors on hand, chest and ankle (accelerometer, gyroscope, magnetometer, temperature) and a heat rate monitor. There are 16 sensors in the dataset sampled at 100Hz.
        \end{itemize}
        
        \begin{table}[h]
            \centering
            \caption{Sensor datasets}
            \label{tab:datasets}
            \resizebox{\columnwidth}{!}{%
            \begin{tabular}{lcccc}
                \toprule
                \textbf{Dataset} & \textbf{Users} & \textbf{Sensor Placement} & \textbf{Sensors} & \textbf{Rate} \\
                \midrule
                \textbf{MHEALTH} & 10 & Chest, R. Wrist, L. Ankle & Acc, Gyro, Mag, ECG & 50 Hz \\
                \textbf{PAMAP2} & 9 & Hand, Chest, Ankle & Acc, Gyro, Mag, Temperature, HR & 100 Hz \\
                \bottomrule
            \end{tabular}%
            }
            
        \end{table}

        \emph{HAR Models and Training Protocol}:
        We employed the TinierHAR model \cite{bib:eval:tinier_har} for our experiments. The model on the smart wearable is a lightweight version of the more complex TinierHAR model executed on the edge server. For all datasets, sensor data was segmented using a 4-second sliding window with 50\% overlap. The data was partitioned into training (80\%) and testing (20\%) sets by subject, so that no subject appears in both partitions and the reported results reflect performance on held-out users. To ensure the edge model remains accurate on the subsampled, reconstructed input it receives at runtime under Viveka's policy, it is trained with rate-degradation augmentation. The training signals are randomly downsampled to rates drawn from the deployment rate grid and reconstructed by linear interpolation, matching the online reconstruction path. This co-designs training with deployment, closing the gap between offline policy computation and online execution. On the wearable, inference is performed periodically according to a decision interval, whose impact on performance and energy is evaluated in our results.
    
    \subsection{Baseline}
    \label{eval:baseline}
        To rigorously evaluate our proposed method, Viveka, we compare it against multiple baselines that span a wide spectrum of approaches. These baselines are systematically designed to represent the state-of-the-art and to isolate the individual contributions of our system's core components: sensor selection (global vs. activity-aware) and sampling rate adaptation (static vs. dynamic).
        \begin{enumerate}[leftmargin=8pt]
            \item \textbf{Oracle}: The Oracle is assumed to have perfect, a priori knowledge of the user's true activity at all times. For any given activity, it activates only the optimal, predetermined set of sensors and operates them at the minimum sampling rate required, as determined by the Nyquist-Shannon theorem. This baseline is not a practical competitor but serves to contextualize performance by establishing the best achievable result.
            \item \textbf{Always On (AlwaysOn)}: This baseline represents a standard, non-adaptive approach. All available sensors are continuously active, operating at the maximum static sampling frequency. It serves as an upper bound for data sensing and energy consumption.
            \item \textbf{Variance-based Sampling (VS)}: In this baseline, all sensors remain active, but their sampling rates are adapted based on signal variance similar to \cite{bib:intro:lasa}. The rate is increased during periods of high variance (indicating activity changes) and decreased during periods of low variance.
            \item \textbf{Global Sensor Selection (GS)}: Only the globally most important subset of sensors, identified by PFI, is activated. This fixed set of sensors operates at the maximum static sampling rate throughout the evaluation which mimics \cite{bib:intro:min_cost, bib:intro:shapley}.
            \item \textbf{Global Sensor Selection with Variance-based Sampling (GS-VS)}: This baseline combines global selection with variance-based sampling. The globally important sensor subset is always active, but its sampling rates are dynamically adjusted based on signal variance \cite{bib:intro:min_cost, bib:intro:adasense}.
            \item \textbf{Activity-based Sensor Selection (AS)}: Based on the inferred activity, a specific subset of sensors is activated. All sensors in this active set operate at a fixed, maximum sampling rate. It is a modified version of \cite{bib:related:idss} with fixed sampling rates.
            \item \textbf{Activity-based Sampling Rate (AR)}: In this approach, all sensors are active. Upon inferring an activity change, the sampling rates for all sensors are increased uniformly. When the activity stabilizes, the rates are gradually decreased \cite{bib:intro:adasense}.
            \item \textbf{Activity-aware Sensor Selection with Variance-based Sampling (AS-VS)}: This approach activates an activity-specific set of sensors and then further optimizes power by adapting their sampling rates uniformly based on real-time signal variance. It is a modified version of \cite{bib:related:idss} with sampling rate variation.
            \item \textbf{Viveka}: Our proposed approach, where sensors are selected based on the inferred activity and their sampling rates are individually adjusted according to the spectral energy analysis.
        \end{enumerate}

        \begin{table}[h]
            \centering
            \caption{Sensors for experiment}
            \label{tab:eval:datasheet}
            \begin{tabular}{||p{2cm} p{3cm} ||} 
                \hline
                Sensor type & Commercial sensor\\ [0.5ex] 
                \hline\hline
                Accelerometer & BMI160 \cite{bib:eval:bmi160} \\
                Gyroscope & BMI160 \cite{bib:eval:bmi160} \\
                Magnetometer & BMM350 \cite{bib:eval:bmm350} \\
                ECG & MAX30003 \cite{bib:eval:max30003} \\
                Temperature & NST112 \cite{bib:eval:nst112} \\
                BLE & nRF52840 \cite{bib:eval:nrf52840}\\
                \hline
            \end{tabular}
        \end{table}
    
    \subsection{Evaluation Metrics}
        We analyze the performance of the proposed method against the baselines using metrics for classification (F1-score), data overhead (data reduction, normalized data collection) and energy consumption (energy savings, normalized energy consumption).
        \begin{itemize}[leftmargin=8pt]
            \item \textbf{Classification Performance}: The primary metric for HAR performance is the F1-score, which is the harmonic mean of precision and recall. Precision measures the accuracy of positive predictions, while recall measures the ability to capture all actual positive instances. We chose the F1-score due to its robustness in handling the class imbalances often present in real-world activity datasets. It is calculated as:
                \begin{align}
                \text{F1-score} = 2 \times \frac{\text{Precision} \times \text{Recall}}{\text{Precision} + \text{Recall}}
                \end{align}
            We compare the F1-score of our method with the baselines to quantify any impact on classification accuracy resulting from our data reduction techniques.
            \item \textbf{Data Reduction}: We quantify data reduction by comparing the total volume of data collected by our approach to that of the All Sensors On baseline. The result is then presented as a normalized value relative to the baseline.
            \item \textbf{Energy Consumption}: We use an energy model that provides a holistic estimate of power consumption by aggregating three key components: sensing, computation, and communication. The total energy, $E_{\text{total}}$, is defined as:
            \begin{align}
            E_{\text{total}} = E_{\text{sense}} + E_{\text{compute}} + E_{\text{comm}}
            \end{align}
            Each component is modeled as follows:
            \begin{itemize}
                \item \textbf{Sensing Energy ($E_{\text{sense}}$)}: This component is calculated using power consumption values from the datasheets of commercially available sensors, as detailed in Table \ref{tab:eval:datasheet}. Our model accounts not only for the energy consumed during active sampling but also for the transactional energy costs of state transitions (i.e., enabling/disabling sensors and changing their sampling rates).
                \item \textbf{Computation Energy ($E_{\text{compute}}$)}: The energy cost of a single HAR model inference is measured empirically. We executed the model for 10,000 consecutive iterations on our target smartwatch (TicWatch Pro 5) and measured the average energy draw using the Android Power Profiler. The energy consumptions are $\approx 15.2$ $\mu$J and $\approx 47.2$ $\mu$J per inference for MHEALTH and PAMAP2 respectively.
                \item \textbf{Communication Energy ($E_{\text{comm}}$)}: The energy required for data transmission is modeled based on specifications from the datasheet for the Nordic Semiconductor nRF52 series Bluetooth Low Energy (BLE) chipset, a common component in modern wearables.
            \end{itemize}
        \end{itemize}
    
    \subsection{Overall Impact of Viveka}
        \begin{figure}[h]
            \centering
            \includegraphics[scale=0.5]{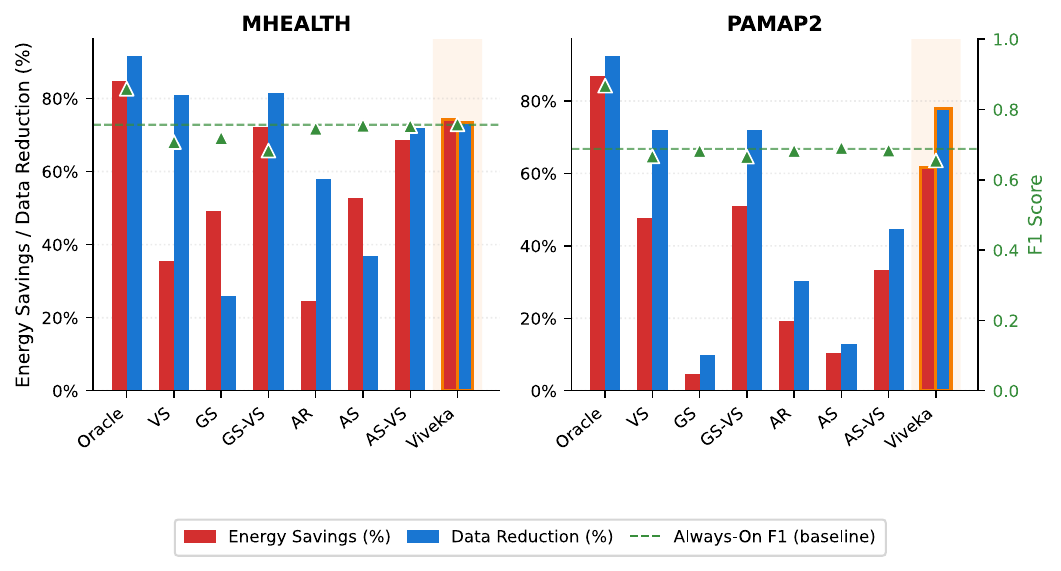}   
            \caption{Viveka achieves the highest energy and data reduction across both datasets, using only ~25\% of Always On energy on MHEALTH and ~36\% on PAMAP2, while maintaining F1 scores comparable to the Always On baseline.}
            \label{fig:eval:overall}
        \end{figure}
        In this experiment, we evaluate the performance of Viveka in terms of energy consumption and data sensed for transmission compared to the baselines discussed in Section \ref{eval:baseline}. The hyperparameters of Viveka are dataset-specific: \{MHEALTH: PFI Threshold: 0.2, FFT Energy Threshold: 0.95, Confidence Threshold: 0.7, Fallback Confidence Threshold: 0.1\} and \{PAMAP2: PFI Threshold: 0.05, FFT Energy Threshold: 0.8, Confidence Threshold: 0.7, Fallback Confidence Threshold: 0.0\}. The values of the parameters were decided based on the parameter tuning experiments discussed in the Section \ref{eval:ptune}. The inference interval is set to 2 seconds for both the datasets.

        Viveka achieves the highest energy and data efficiency among all non-oracle strategies on both datasets. On MHEALTH, Viveka reduces energy consumption by 74\% and data transmission by 73\% relative to Always On, while maintaining the F1 score close to Always On baseline. On PAMAP2, Viveka reduces energy by 62\% and data by 78\%. Among the competing strategies, GS-VS approaches Viveka's energy savings on MHEALTH (72\% vs. 74\%) but does so at a significantly higher F1 cost (0.683 vs. 0.758). Strategies that aggressively reduce data transmission (VS, GS-VS) consistently sacrifice classification accuracy, while activity-aware strategies (AR, AS, AS-VS) recover some accuracy but at substantially lower efficiency.

        \textbf{Takeaway}: Viveka uniquely balances efficiency and accuracy. It achieves the largest reduction in energy and data transmitted while keeping F1 within acceptable bounds of the Always On baseline. It demonstrates that context-aware, tiered decision making outperforms both traditional duty-cycling and static activity based policies.

        \textbf{Operational longevity}: To evaluate the impact of Viveka on operational longevity, we use the energy consumption profiles derived from both datasets. We assume the total energy budget is dominated by sensing, inference computation and BLE transmission. While real-world deployments involve additional constant energy factors, we focus on these primary components for simplicity. We assume a total battery capacity of 300 mAh, which is representative of modern smartwatches. The activity data available per user in the MHEALTH and PAMAP2 datasets used for testing span 38 minutes and 71 minutes, respectively, on average. Based on the power consumption analysis of the experimental results shown in Figure \ref{fig:eval:overall}, Table \ref{tab:eval:energy} projects the operational life of the BSN on a single charge. The results indicate that the AlwaysOn configuration, requires recharging approximately every 8.7(MHEALTH) or 4.1(PAMAP2) days. This short duration is primarily driven by the high duty-cycle of the sensors and the continuous data stream over BLE. In contrast, the Viveka strategy achieves a projected battery life of 34.6(MHEALTH) or 10.8(PAMAP2) days. This improvement demonstrates the effectiveness of context-aware sensor gating. The difference in the number of days for Viveka between MHEALTH and PAMAP2 stems from two factors: (1) PAMAP2's sensor suite (16 sensors) is bigger than MHEALTH suite (8 sensors). In addition PAMAP2 is having a higher maximum sample rate (100Hz) compared to MHEALTH (50Hz). (2) On MHEALTH, Viveka aggressively turns off gyroscopes, whereas in PAMAP2, they are mostly relevant. The higher sampling rate and sensor count results in BLE transmission energy to be elevated. Together these factors directly translate to the observed difference in longevity.
        \begin{table}[htbp]
            \centering
            \caption{Projected BSN Operational Longevity using a 300 mAh Battery for MHEALTH and PAMAP2}
            \label{tab:eval:energy}
            \footnotesize
            \setlength{\tabcolsep}{2.5pt}
            \begin{tabular}{@{}lcccc@{}}
                \toprule
                \textbf{Strategy} & \multicolumn{2}{c}{\textbf{38-Min Activity}} & \multicolumn{2}{c}{\textbf{71-Min Activity}} \\
                \cmidrule(lr){2-3} \cmidrule(lr){4-5}
                & \textbf{\shortstack{Power\\(mW)}} & \textbf{\shortstack{Life\\(Days)}} & \textbf{\shortstack{Power\\(mW)}} & \textbf{\shortstack{Life\\(Days)}} \\
                \midrule
                AlwaysOn             & 5.30 & 8.7  & 11.27 & 4.1  \\
                VS                   & 3.42  & 13.5 & 5.89 & 7.8  \\
                GS                   & 2.69  & 17.2 & 10.77 & 4.3  \\
                GS-VS                & 1.47  & 31.4 & 5.51 & 8.4  \\
                AR                   & 3.99  & 11.6 & 9.10 & 5.1  \\
                AS                   & 2.50  & 18.5 & 10.11 & 4.6  \\
                AS-VS                & 1.66  & 27.8 & 7.53 & 6.1  \\
                Oracle (theoretical) & 0.81  & 56.9 & 1.48  & 31.2 \\
                \midrule
                \textbf{Viveka} & \textbf{1.34} & \textbf{34.6} & \textbf{4.29} & \textbf{10.8} \\
                \bottomrule
            \end{tabular}
        \end{table}

    \subsection{Parameter Tuning}
    \label{eval:ptune}
        \begin{figure*}
            \centering
            \includegraphics[scale=0.5]{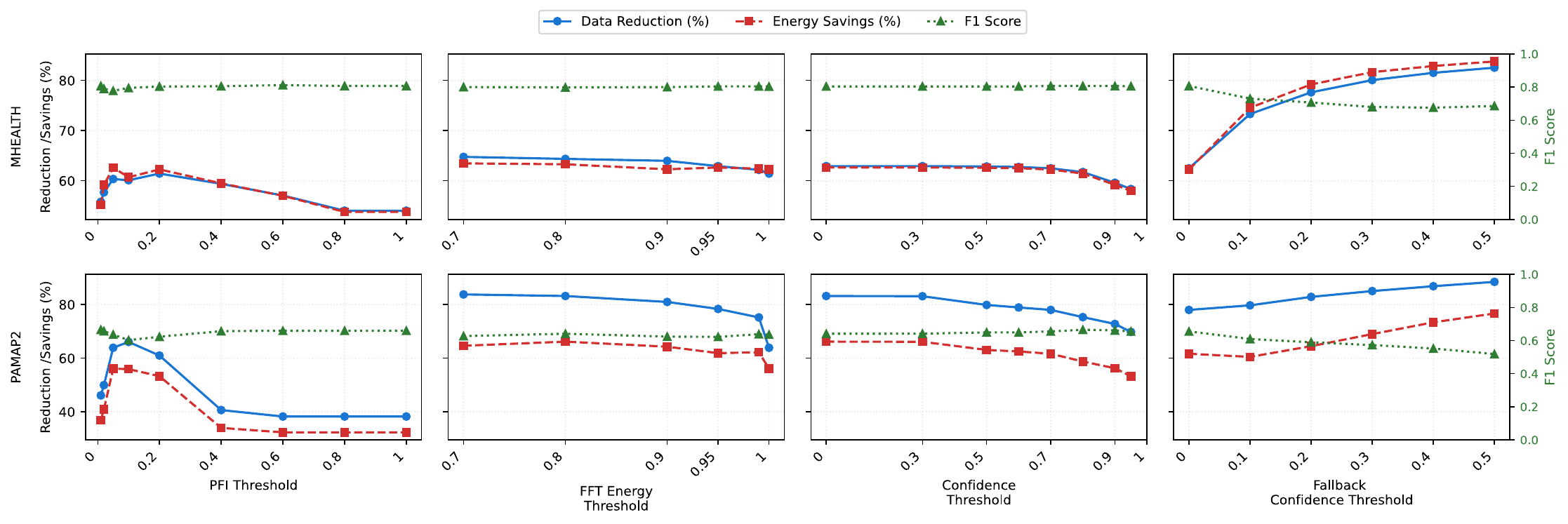}   
            \caption{Parameter Tuning. }
            \label{fig:eval:ptune}
        \end{figure*}
        Viveka performance is controlled by four hyperparameters: PFI Threshold, FFT Energy Threshold, Confidence Threshold and Fallback Confidence Threshold. In this section, we will explore tuning each parameter for both the datasets. Each parameter is optimized in isolation, while the remaining parameters are held fixed. At the start, all the parameters are at their neutal values (FFT Energy threshold: 1.0, Confidence threshold: 0.0, Fallback Confidence threshold: 0.0), which effectively disables each mechanism. Once the best value for a parameter is identified, it is locked in before the next parameter tuning begins. This approach reduces the combinatorial search space to four independent sweeps, with each sweep building on the gains of the previous one. The results are shown in Figure \ref{fig:eval:ptune}.
        \subsubsection{PFI Threshold}
            The PFI threshold ($\theta_{PFI}$) controls which sensors are retained in the per-activity policy, lower values retain more sensors. The $\theta_{PFI}$ is varied across nine discrete values ranging from 0.01 to 1.0 \{0.01, 0.02, 0.05, 0.1, 0.2, 0.4, 0.6, 0.8, 1.0\}.  On MHEALTH, data reduction and energy savings increase as $\theta_{PFI}$ rises from 0.01 to 0.2, peaking at 61.4\% and ~62\% respectively, as fewer sensors are retained and less data is transmitted. Beyond $\theta_{PFI}=0.2$, both metrics decline as over pruning forces more conservative fallback policies. F1 score follows a non-monotonic trend, dipping at intermediate values ($\theta_{PFI}=0.05$, F1=0.777) before recovering to 0.803 at $\theta_{PFI}=0.2$. This trend suggests that moderate pruning removes noisy sensors without sacrificing classification accuracy. On PAMAP2, energy savings and data reduction peak earlier at $\theta_{PFI}=0.05$ (63.9\% data reduction, ~56\% energy savings), with a sharp F1 drop to 0.603 at $\theta_{PFI}= 0.1$. The values above 0.2 degrade all three metrics simultaneously as pruning eliminates sensors critical to activity discrimination.
            
            \textbf{Takeaway}: A very low threshold keeps nearly all sensors, while an strict threshold eliminates sensors that are still relevant, forcing the system into fallback modes that activate a broader sensor set, undermining the intended savings. The resulting non-monotonic curve reflects this tension: MHEALTH settles at $\theta_{PFI}=0.2$, where the pruned sensor subset is small enough to reduce sensing and BLE energy while keeping the system in its high-certainty, low-sensor mode enough to maintain F1. PAMAP2's richer and more variable activity space requires a stricter cutoff ($\theta_{PFI}=0.05$) to eliminate sufficient redundant sensors. Beyond that point, critical sensors are dropped, activity recognition degrades and the system spends more time in the higher-sensor fallback state.

        \subsubsection{FFT Energy Threshold}
            The FFT Energy threshold ($\theta_{FFT}$) controls the sample rates of sensors. For the tuning of $\theta_{FFT}$, the $\theta_{PFI}$ is locked at its best values (MHEALTH: 0.2, PAMAP2: 0.05). The $\theta_{FFT}$ is varied across six discrete values ranging from 0.7 to 1.0 \{0.7, 0.8, 0.9, 0.95, 0.99, 1.0\}. On MHEALTH, all values below 1.0 improve data reduction and energy savings over the reference ($\theta_{FFT}$=1.0: 61.4\% data, ~62\% energy), with $\theta_{FFT} = 0.7$ achieving the highest data reduction (64.8\%) and lowest energy. However, F1 drops below the reference for all values except $\theta_{FFT} =0.95$ (F1=0.8037), which marginally improves accuracy while still gaining in data reduction and energy savings. On PAMAP2, FFT-based rate reduction delivers dramatically larger gains across all metrics. All the values reduce energy from ~54\% to ~66\% and data reduction from 63.9\% to 84\%. $\theta_{FFT}=0.8$ achieves the best balance, with 83.2\% data reduction, 66.2\% energy savings and the highest F1 among all FFT values (0.641), while $\theta_{FFT}=0.7$ marginally improves data reduction (83.8\%) but at a lower F1 (0.627).
            
            \textbf{Takeaway}: FFT Energy threshold is the primary lever for controlling data volume. In PAMAP2, lowering the threshold from 1.0 to 0.8 cuts BLE transmission energy from 7,412 mJ to 3,463 mJ by reducing the sample rate alone, with sensor count unchanged. MHEALTH shows a more modest but consistent gain (BLE drops from 1,216 mJ to 1,170 mJ at the chosen $\theta_{FFT}=0.95$), since its activities are less rate-sensitive.

        \subsubsection{Confidence Threshold}
            The Confidence threshold ($\theta_{conf}$) controls the activation of the context aware policy in Viveka. The $\theta_{PFI}$ and $\theta_{FFT}$ values are locked in ($\theta_{PFI}$ (MHEALTH: 0.2, PAMAP2: 0.05), $\theta_{FFT}$(MHEALTH: 0.95, PAMAP2: 0.80)). The $\theta_{conf}$ is varied across nine values ranging from 0.0 to 1.0 \{0.0, 0.3, 0.5, 0.6, 0.7, 0.8, 0.9, 1.0\}. On MHEALTH, energy savings and data reduction decrease modestly as the threshold increases as fewer windows qualify for the optimized low-rate policy. For F1, the improvement is negligible. On PAMAP2, the same tradeoff is more pronounced: data reduction falls from 83.2\% to 69.8\%, and energy savings decline from 64\% to 55\%, while F1 rises slightly from 0.641 to a peak of 0.665.
            
            \textbf{Takeaway}: Confidence threshold controls how frequently the system selects the reduced-sensor, reduced-rate granular policy versus deferring to a multi-sensor fallback. Too low a threshold applies the optimized low-sensor, low-rate settings even on uncertain windows, occasionally feeding the classifier degraded or incomplete feature vectors and lowering F1. Too high a threshold keeps the system in its conservative fallback state, where more sensors remain active at default rates. $\theta_{conf} = 0.7$ is the optimal value for both datasets: it ensures the reduced-sensor, lower-rate policy is applied only under confidence, yielding a small but consistent F1 gain at the cost of marginally higher energy use.

        \subsubsection{Fallback Confidence Threshold}
            The Fallback confidence threshold ($\beta_{conf}$) controls sensor activations during uncertainity. The $\theta_{PFI}$ ,$\theta_{FFT}$ and $\theta_{conf}$ values are locked in ($\theta_{PFI}$ (MHEALTH: 0.2, PAMAP2: 0.05), $\theta_{FFT}$(MHEALTH: 0.95, PAMAP2: 0.80), $\theta_{conf}$(MHEALTH: 0.7, PAMAP2: 0.7)). The $\beta_{conf}$ is varied across six values ranging from 0.0 to 0.5 \{0.0, 0.1, 0.2, 0.3, 0.4, 0.5\}. On MHEALTH, increasing threshold from 0.0 progressively improves both energy savings and data reduction, from 62.5\% to 82.6\%, as more windows are routed through the combined-sensor policy at full rate rather than the expensive intermediate fallback. F1 declines steadily from 0.8065 to 0.685, with $\beta_{conf} = 0.1$ offering the best tradeoff (73.3\% data reduction, 74.6\% energy savings, F1=0.730). On PAMAP2, the threshold has opposite effect: $\beta_{conf}=0.1$ simultaneously worsens energy savings (61.6\% $\rightarrow$ 60.5\%) and reduces F1 (0.655 $\rightarrow$ 0.609), as the combined-sensor policy at high sampling rate incurs more overhead than it saves.
            
            \textbf{Takeaway}: The Fallback confidence threshold determines which candidate activities contribute sensors during uncertainty windows. With $\beta_{conf}=0.0$, all activities pass the softmax threshold (every probability is strictly greater than zero), so the union of all activity-specific sensors is activated at high sampling rate whenever high certainty is not met, forming a broad safety net at high sensor and rate cost. Raising $\beta_{conf}$ restricts this to only the most probable candidate activities, shrinking the active sensor set and proportionally reducing both sensing energy and transmitted data. In MHEALTH, $\beta_{conf} = 0.1$ narrows the sensor pool sufficiently to yield ~12 percentage points of additional energy savings (62\% $\rightarrow$ 75\%) and 11 percentage points of data reduction, though the F1 dropped from 0.806 to 0.730. In PAMAP2, even a small $\beta_{conf}$ value overfilters sensors needed for its broader activity vocabulary, so $\beta_{conf} = 0.0$ is retained to preserve F1, relying instead on Confidence threshold driven savings.

    \subsection{Impact of Decision Interval}
        \begin{figure}
            \centering
            \includegraphics[scale=1]{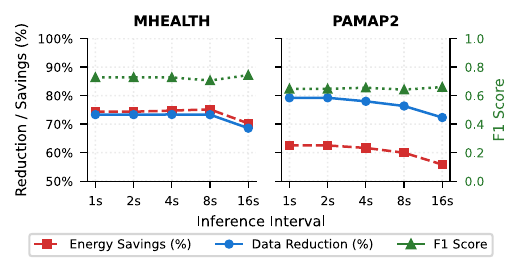}   
            \caption{Decision interval. As decision interval increases, energy savings and data reduction decreases for both datasets.}
            \label{fig:eval:dec_interval}
        \end{figure}

        The decision interval ($T_{interval}$) tuning is conducted with all algorithm parameters fixed at their tuned values: for MHEALTH, $\theta_{PFI}=0.2$, $\theta_{FFT}=0.95$, $\theta_{conf}=0.7$, $\beta_{conf}=0.1$; for PAMAP2, $\theta_{PFI}=0.05$, $\theta_{FFT}=0.8$, $\theta_{conf}=0.7$, $\beta_{conf}=0.0$. The $T_{interval}$ is varied over [1, 2, 4, 8, 16] seconds for both datasets (Figure \ref{fig:eval:dec_interval}).

        For MHEALTH, intervals of 1s and 2s produce identical results (energy savings=74.4\%, data reduction=73.3\%, F1=0.728), with marginal further gains at 4s and 8s (savings reaching 75.2\%). At 16s, performance degrades noticeably, energy savings drop to 70.3\% and data reduction to 68.6\%, as the system fails to keep up with MHEALTH's frequent activity transitions. For PAMAP2, the trend is monotonically degrading: energy savings fall from 62.5\% at 1s to 55.8\% at 16s, while data reduction drops from 79.2\% to 72.3\%. F1 remains relatively stable across the sweep for both datasets.

        \textbf{Takeaway}: The primary cost of a longer decision interval is not reduced inference energy, which is negligible relative to sensing and transmission, but rather stale sensor-rate policies. When policy updates are infrequent, the system holds its current sensor subset and sample rate configuration beyond the activity boundary, collecting and transmitting more data than necessary until the next decision point. Shorter intervals enable tighter policy tracking, keeping the active sensor count and sample rates aligned with the current activity. A 2s interval is suitable for both MHEALTH and PAMAP2 based on the observed results.

    \subsection{Sensor Activations}
        \begin{figure*}[h]
            \centering
            \resizebox{0.7\textwidth}{!}{
                \includegraphics[scale=0.7]{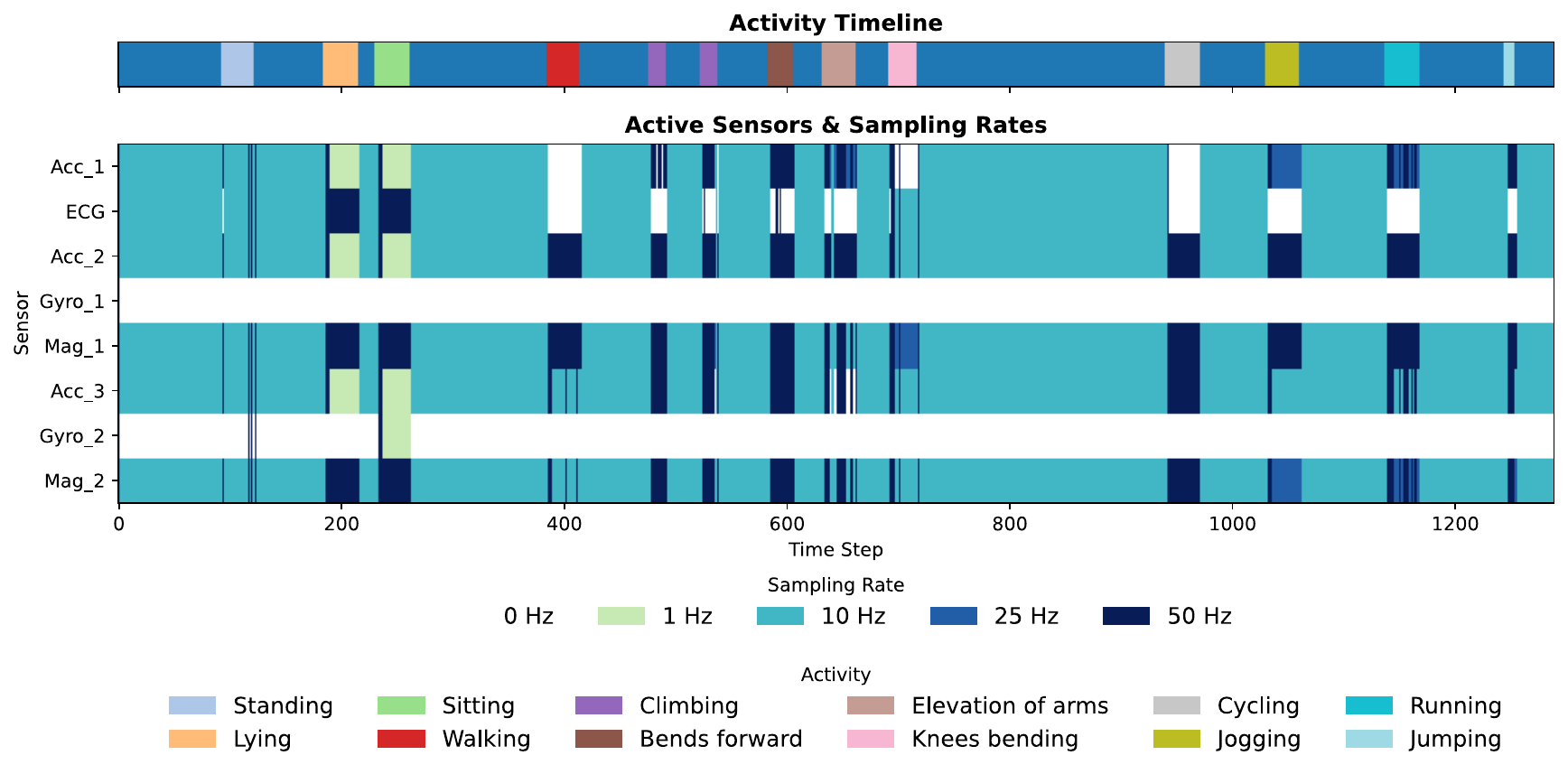}   
            }
            \caption{Sensor activations for MHEALTH. The top ribbon displays the ground truth activity timeline, while the heatmap below illustrates the dynamic sensor selection and sampling rate adaptation (0-50 Hz) triggered by the Viveka framework. It highlights the system's ability to selectively prune sensors and scale sampling rate according to activity intensity.}
            \label{fig:eval:activation}
        \end{figure*}

        In this analysis, we examine the granular sensor activations of the MHEALTH dataset, configured with the optimal parameters identified in the previous sections. The sensor activation map (Figure \ref{fig:eval:activation}) provides a granular view of the Viveka framework's runtime behavior, confirming its ability to decouple sensing cost from continuous monitoring.
        
        The transitions in the sensor heatmap (bottom) align perfectly with the activity transitions in the ground truth ribbon (top). Even Transient periods (brief transitions between activities) are captured with distinct sensor configurations. The system does not merely classify steady states, it is highly responsive to state changes, correctly identifying the start and end of specific movements without significant lag.

        The rows corresponding to Gyroscopes remain predominantly white (0 Hz) or light green (1 Hz) across most activities. They are only activated (darker colors) during specific complex tasks. This confirms that the most power-hungry sensors are often redundant for recognizing common human activities. The system successfully prunes them, relying instead on cheaper modalities (Accelerometers), thereby validating the PFI feature selection strategy.

        There is a clear correlation between activity intensity and color saturation. High-intensity activities (e.g., Jogging, Running) trigger dark navy blocks (50 Hz), while low-intensity activities (e.g., Sitting, Standing) trigger light green/teal blocks (1-10 Hz). Viveka exhibits context-aware scalability. It dynamically scales the sampling rates proportional to the activity rigor, ensuring no energy is wasted oversampling stationary behavior.

        The vertical slices of the heatmap show that for any given time step, only a fraction of the available sensors are active (non-white). Very rarely is the entire column lit up. Viveka achieves sparsity, proving that a full sensor suite is almost never required simultaneously. Each activity requires only a minimal subset of sensors.

        Within continuous activity blocks, the color patterns remain consistent (solid blocks of color rather than flickering noise). The strategy demonstrates stability. It avoids sensor thrashing (rapidly toggling sensors On/Off), which is critical because the energy overhead of powering up a sensor often outweighs the savings of turning it off for a split second.

        \textbf{Takeaway}: The visualization confirms that Viveka achieves efficiency through context-aware semantic adaptation, replacing static sampling schedules with dynamic, stable and minimalist sensing signatures tailored to the instantaneous activity rigor of the user.

%% file: sections/related.tex
\section{Related Work}
    Energy efficiency is a paramount concern in wearable Human Activity Recognition (HAR), driving research into resource-aware methodologies. The literature has largely pursued this goal through two primary avenues: sensor selection and adaptive sampling, with the most advanced works combining the two.

    Sensor Selection Strategies: 
    Sensor selection methods aim to reduce power consumption by activating only the most informative sensors. Static approaches pre-determine an optimal sensor set using techniques like game theory \cite{bib:intro:shapley}, bio-inspired algorithms \cite{bib:intro:glowworm}, or cost optimization \cite{bib:intro:min_cost}. However, their fixed nature limits adaptability to dynamic user states. Dynamic methods offer more flexibility by using Bayesian models to quantify uncertainty (VFDS) \cite{bib:intro:vfds} or employing specialized architectures that can internally weigh sensor importance (DANA) \cite{bib:intro:dana}. The primary drawback of these advanced techniques is that they often require complex, model-specific architectures, sacrificing the modularity needed to work with general-purpose HAR models.

    Adaptive Sampling Techniques:
    Adaptive sampling techniques focus on reducing data volume by modulating the sampling frequency. These can be context-aware \cite{bib:intro:context_aware}, using a lightweight model to infer the activity and adjust rates accordingly (FreqSense) \cite{bib:intro:freqsense}, or purely data-driven, reacting directly to signal properties like variance (LASA-IoT) \cite{bib:intro:lasa}. A key limitation is that they often apply a uniform sampling rate to all active sensors, failing to account for the fact that different sensors have different information content and power costs, leading to suboptimal energy savings.

    Co-optimization of Selection and Sampling: 
    Recognizing that these problems are coupled, the most advanced research seeks to co-optimize both selection and sampling. A prevalent paradigm is the hierarchical system, where a low-power configuration detects a potential change, triggering a more resource-intensive mode for precise identification \cite{bib:intro:wear_energy}. While effective, these systems often act as a simple binary switch, lacking the granularity to modulate individual sampling rates or select specific sensor subsets during transitions. More sophisticated frameworks like CoSS \cite{bib:intro:coss} use reinforcement learning to learn a joint policy. However, this approach is model-integrated, requiring specialized training and resulting in a static configuration that cannot adapt to real-time dynamics post-deployment.

%% file: sections/conclusion.tex
\section{Conclusion}
    The full potential of smart wearables hinges on extending their battery life, as frequent recharging creates data gaps and user inconvenience. To address this, we presented Viveka, a context-aware sensing framework designed to dramatically improve both energy and data efficiency. Viveka adapts to the user's context by using permutation-based feature importance to select the most relevant sensors and the spectral energy analysis to determine their optimal sampling rates for the current activity. Our framework also includes robust fallback mechanisms to ensure high accuracy is maintained, even when on-device context detection is uncertain.

    Our extensive evaluation on the MHEALTH and PAMAP2 datasets demonstrates that Viveka achieves the best tradeoff between accuracy and energy efficiency compared to existing strategies. The framework achieves upto 75\% energy savings and 78\% data reduction compared to standard baselines, while maintaining classification accuracy within 3-5\% of the baselines. We also conducted a sensitivity analysis across all the parameters used in the framework to identify the sweet spot.

    Currently, Viveka relies on offline sensor selection to define sensor subsets. The future iterations will explore online learning to dynamically adjust feature importance thresholds in real-time. We plan to extend Viveka from a single-device optimization to a distributed multi-wearable ecosystem. By coordinating sensor activation across multiple devices (e.g., smartwatch, smartphone, earbuds), the framework could dynamically offload sensing tasks to the device with the highest residual battery, maximizing the collective lifespan of the user's Body Area Network (BAN). Our current energy metrics are derived from validated sensor datasheets. Future work will involve deploying Viveka on ultra low power microcontrollers (e.g., ARM Cortex-M4) to quantify the precise interplay between sensor activations and the associated hardware load.

%% file: sections/appendix.tex
\section{NP-Hardness Proof} 
\label{apdx:nphard}
    \begin{theorem}
        The APSSE problem formulated in equations \eqref{eq:min:obj}--\eqref{eq:cons:3} is NP-Hard.
    \end{theorem}

    \begin{proof}
        We prove this by reduction from the \textbf{Weighted Set Cover (WSC)} problem, which is known to be NP-Hard.

        \textbf{The Weighted Set Cover Problem:}
        Given a universe of elements $U = \{u_1, u_2, \dots, u_m\}$ (elements to be covered) and a collection of subsets $\mathcal{S} = \{S_1, S_2, \dots, S_n\}$ where each $S_i \subseteq U$ has an associated cost $w_i$. The goal is to select a sub-collection of sets such that their union covers all elements in $U$ and the total cost is minimized.

        \textbf{Mapping to APSSE:}
        We construct a special instance of the APSSE problem that is isomorphic to the WSC instance:

        \begin{enumerate}
            \item Mapping Activities to Universe: Let the Universe $U$ represent the set of activity types required. The condition $C(x) \leq C_{threshold}$ is equivalent to covering all necessary activity types (i.e., achieving 100\% required accuracy implies covering all elements in $U$).
            \item Mapping Sensors to Sets: Let each sensor $s$ represent a set $S_i$ in the WSC problem.
            \item Mapping Sampling Rates to Selection: Restrict each sensor $s$ to have only two rates: $R_s = \{r_{s,0}, r_{s,1}\}$ where $r_{s,0}=0$ and $r_{s,1}=r_{high}$ for simplicity.
                \begin{itemize}
                    \item Choosing index $j=0$ (OFF) corresponds to \textit{not selecting} set $S_i$. Cost = 0.
                    \item Choosing index $j=1$ (ON) corresponds to \textit{selecting} set $S_i$. Cost = $E(s, r_{s,1}) = w_i$.
                \end{itemize}
            \item Mapping Accuracy to Coverage: We define the error function $C(x)$ such that $C(x) \leq C_{threshold}$ is satisfied \textit{if and only if} the union of features provided by the selected active sensors fully covers the universe $U$. If any element $u \in U$ is uncovered, the error remains above the threshold.
        \end{enumerate}

        Conclusion: Finding the minimum energy configuration for APSSE in this instance is equivalent to finding the minimum weight collection of sets that covers $U$. Since Weighted Set Cover is NP-Hard and it can be reduced to a specific instance of APSSE, the general APSSE problem is also NP-Hard.
    \end{proof}

    \begin{remark}
    The formulated APSSE problem addresses the selection of a static configuration $x$. The problem of per-activity adaptive selection (finding a mapping from activities to configurations) is a generalization of this static case. Since the static problem (where the mapping is constrained to be constant) is NP-Hard, the general per-activity optimization problem is also NP-Hard.
    \end{remark}